\documentclass[10pt]{article}

\usepackage{amsxtra,amssymb,amsthm,amsmath,latexsym}

\usepackage{graphicx}
\newtheorem{thm}{Theorem}[section]

\newtheorem{lem}[]{Lemma}

 \newcommand{\thmref}[1]{Theorem~\ref{#1}}
 \newcommand{\lemref}[1]{Lemma~\ref{#1}}

\newcommand{\R}{{\mathbb R}}
\newcommand{\C}{{\mathbb C}}

\newcommand{\bee}{\begin{equation*}}
\newcommand{\eee}{\end{equation*}}
\newcommand{\be}{\begin{equation}}
\newcommand{\ee}{\end{equation}}
\newcommand{\pn}{\par\noindent}

\title{Dynamical Systems Method (DSM) for solving nonlinear operator
equations in Banach spaces}
\author{A G Ramm\\
\small Department of Mathematics\\[-0.8ex]
\small Kansas State University, Manhattan, KS 66506-2602, USA\\[-0.8ex]
\small \texttt{ramm@math.ksu.edu}\\
}

\begin{document}

Eurasian Math. Journal, 3, N1, (2012), 5-19.

 \date{} \maketitle \begin{abstract} Let $F(u)=h$ be a solvable operator
equation
in a Banach space $X$ with a Gateaux differentiable norm. Under minimal
smoothness assumptions on $F$, sufficient conditions are given for
the validity of the Dynamical Systems Method (DSM) for solving the above
operator equation.
It is proved that the DSM (Dynamical Systems Method) \bee
\dot{u}(t)=-A^{-1}_{a(t)}(u(t))[F(u(t))+a(t)u(t)-f],\quad u(0)=u_0,\
\eee converges to $y$ as $t\to +\infty$, for
$a(t)$  properly chosen. Here  $F(y)=f$, and $\dot{u}$ denotes the time
derivative.
\end{abstract}
\pn{\\ MSC 2000, 47J05, 47J06, 47J35  }

\section{Introduction} Consider an operator equation \be\label{e1} F(u)=f,
\ee where $F$ is an operator in a Banach space $X$
 with a Gateaux-differentiable norm. Assume that $F$ is continuously
Fr\'{e}chet differentiable, $F'(u):=A(u)$. Denote by $A_a:=A+aI$,
where $I$ is the identity operator, and by $c_j$, $j=0,1,2,3$, various
positive constants. Let $L$ be  a smooth path on the
complex plane $\C$ joining the origin and some point $a_0$,
$0<|a_0|<\epsilon_0$, where
$\epsilon_0>0$ is a small fixed number independent of $u$.

The following assumptions {\bf A1-}
{\bf A3} are valid throughout the paper.

{\bf A1.} {\it Assume that
\be\label{e2} \|A(u)-A(v)\|\leq c_0\|u-v\|^\kappa,
\quad \kappa\in (0,1], \ee
where $\kappa$ is a constant.}

{\bf A2.} {\it Assume that} \be\label{e3} \|A_a^{-1}(u)\|\leq
\frac{c_1}{|a|^b};\,\,\, \forall a\in L,\,\,\, 0<|a|<\epsilon_0.\ee

Assumption \eqref{e3} holds if there is a smooth path $L$ on a complex
$a$-plane, consisting of regular points of the operator $A(u)$, such
that the norm of the resolvent $A_a^{-1}(u)$ grows, as $a\to 0$,
not faster than a power $|a|^{-b}$. Thus, assumption \eqref{e3}
is a weak assumption. For example, assumption \eqref{e3} is satisfied for
the class of linear operators $A$, satisfying {\it the spectral
assumption,}
introduced in \cite{R499}, Chapter 8. This spectral assumption
says, that the set $\{a: |\arg a-\pi|\leq \phi_0, \,\, 0<|a|<\epsilon_0\}$
consists of the regular points of the operator $A$. This assumption
implies the estimate $||A_a^{-1}||\leq \frac{c_1}{a},\,\, 0<a<\epsilon_0,$
that is, estimate \eqref{e3} with $b=1$ and $a\in (0, \epsilon_0)$.

{\bf A3.}   {\it Assume  that the equation \be\label{e4}
F(w_a)+aw_a-f=0,\quad a\in L, \ee is uniquely solvable for any $f\in
X$, and \be\label{e5} \lim_{a\to 0,a\in L}\|w_a-y\|=0,\quad F(y)=f.
\ee }
 {\it Assume that there exists a constant $c>0$ such that}
\be\label{e5'} |\dot{a}(t)|\le c|\dot{r}(t)|, \qquad r(t):=|a(t)|.
\ee

In formula \eqref{e13} (see below) inequality $|\dot{r}(t)|\le
|\dot{a}(t)|$
is established. Thus,
\be\label{en}
|\dot{r}(t)|\le |\dot{a}(t)|\le c|\dot{r}(t)|.
\ee

We formulate the main result at the end of the paper for convenience of
the reader, because some additional assumptions, used in the proof of
Theorem 2.1 are flexible and will arise naturally in the course of the
proof.

{\it One of the goals in this paper is to demonstrate the methodology
for establishing the convergence results of the type obtained in Theorem
2.1. }

All our assumptions are satisfied, for example, if $F$ is a monotone
operator in a Hilbert space $H$ and $L$ is a segment $[0,\epsilon_0]$.
In this case $c_1=1$ and $b=1$. Our assumptions are satisfied for
the class of operators satisfying a spectral assumption, mentioned above,
which was studied in
\cite{R499} in connection to the Dynamical System Method (DSM)
for solving operator equations.
Sufficient conditions
for \eqref{e5} to hold are given in \cite{R499}.

Every equation \eqref{e1} with a linear, closed, densely defined in
a Hilbert space $H$
operator $F=A$ can be reduced to an equation with a monotone operator
$A^*A$, where $A^*$ is the adjoint to $A$. The operator $T:=A^*A$ is
selfadjoint and densely defined in $H$.  If $f\in D(A^*)$, where $D(A^*)$
is the domain of $A^*$, then the equation $Au=f$ is equivalent to
$Tu=A^*f$, provided that $Au=f$ has a solution, i.e., $f\in R(A)$, where
$R(A)$ is the range of $A$.
Recall that $D(A^*)$ is dense in $H$ if $A$ is closed and densely
defined in $H$.
If $f\in R(A)$ but $f\not\in D(A^*)$, then equation $Tu=A^*f$ still makes
sense and its normal solution $y$, i.e., the solution with minimal
norm, can be defined as
\be\label{e6} y=\lim_{a\to 0}T_a^{-1}A^*f.
\ee
One proves that $Ay=f$, and $ y\perp N(A),$
where $N(A)$ is the null-space of $A$. These
results are proved in \cite{R499}.

Our aim is to prove convergence of the DSM
for solving equation \eqref{e1}:
\be\label{e7}
\dot{u}(t)=-A_{a(t)}^{-1}(u(t))[F(u(t))+a(t)u(t)-f],\quad u(0)=u_0,
\ee where
$u_0\in X$ is an initial element, $a(t)\in
C^1[0,\infty)$, $a(t)\in L$. The DSM version \eqref{e7} is a
computationally efficient analog of a continuous regularized Newton's
method for solving equation \eqref{e1}. Other versions of DSM are
studied in \cite{R499}. In \cite{R606} an approach to a justification
of the DSM in Banach spaces is developed. The ideas from \cite{R606} are
used in this paper. Among other things, an important Lemma 1 is
formulated in a more general form than in \cite{R499}, see also
 \cite{R596},\cite{R601}, \cite{R604},  \cite{R605},  \cite{R611},
\cite{R612}.
Our main result is formulated in \thmref{thm1}, in Section 2.

The DSM for solving operator equations has been developed in the
monograph \cite{R499}. It was used as an efficient computational
tool in \cite{R574}, \cite{R612}. One of the earliest papers on the
continuous analog of Newton's method for solving well-posed
nonlinear operator equations was \cite{G}.

The novel points in our paper include the larger class of
the operator
equations than earlier considered, and the weakened assumptions on the
smoothness of the nonlinear operator $F$: in \cite{R499}
it was often assumed that $F^{\prime \prime}(u)$ is locally bounded, in
the current paper
a much weaker assumption \eqref{e2} is used.

Our proof of \thmref{thm1}
uses the following result.
\begin{lem}\label{lem1} Assume
that $g(t)\geq 0$ is continuously differentiable on any interval $[0,T)$,
on which it is defined, and satisfies the following inequality:
\be\label{e8} \dot{g}(t)\leq
-\gamma(t)g(t)+\alpha(t,g)+\beta(t),\quad
t\in[0,T), \ee where  $\alpha(t,g)$, $\gamma(t)$ and
$\beta(t)$ are rel-valued continuous on $[0,\infty)$ functions of $t$,
$\alpha(t,g)$ is locally Lipschitz with respect to $g$.
Suppose that there
exists a $\mu(t)>0$, $\mu(t)\in C^1[0,\infty)$, such that \be\label{e9}
\alpha(t,\mu^{-1}(t))+\beta(t)\leq
\mu^{-1}(t)[\gamma(t)-\dot{\mu}(t)\mu^{-1}(t)],\qquad t\geq 0, \ee and
\be\label{e10} \mu(0)g(0)\le 1. \ee Then $T=\infty$, i.e., $g$ exists on
$[0,\infty)$, and
\be\label{e11} 0\leq g(t)\le \mu^{-1}(t),\quad t\geq 0.
\ee

\end{lem}
This lemma generalizes  some results from \cite{R499}, \cite{R570}.
It is useful in a study of large-time behavior of solutions to
evolution problems, which are important in many appications, see,
for example, \cite{DK}, \cite{T}, \cite{R605}.  Lemma 1 is proved at
the end of the paper for convenience of the reader and for making
this paper essentially self-contained. We apply Lemma 1 with
$\alpha(t,g)=\alpha(t)g^p$, $p>1$ is a constant, and $\alpha(t)>0$
is a continuous function.

In Section 2 a method is given for a proof of the following conclusions:

{\it There exists a unique solution $u(t)$ to problem \eqref{e7} for all
$t\geq
0$, there exists $u(\infty):=\lim_{t\to \infty}u(t)$, and
$F(u(\infty))=f$, that is:}
\be\label{e12} \exists ! u(t)\quad \forall t\geq 0; \ \exists
u(\infty);\quad F(u(\infty))=f. \ee

The assumptions on $u_0$ and $a(t)$ under which conclusions \eqref{e12}
hold for the
solution to problem \eqref{e7} are formulated in \thmref{thm1} in Section
2.
Theorem 2.1 in  Section 2 is our main result. Roughly speaking,
this result says that conclusions \eqref{e12} hold for the solution
to problem \eqref{e7}, provided that $a(t)$ is suitably chosen.

\section{Proofs} Let $|a(t)|:=r(t)>0$. If $a(t)=a_1(t)+ia_2(t)$, where
$a_1(t)=\text{Re}\, a(t),\, a_2(t)=\text{Im}\, a(t)$, then
\be\label{e13}
|\dot{r}(t)|\leq |\dot{a}(t)|. \ee
Indeed, \be\label{e14}
|\dot{r}(t)|=\frac{|a_1\dot{a}_1+a_2\dot{a}_2|}{r(t)}\leq
\frac{r(t)|\dot{a}(t)|}{r(t)}, \ee and \eqref{e14} implies \eqref{e13}.

Let \be\label{e15} g(t):=\|z(t)\|,\qquad z(t):=u(t)-w_a(t), \ee
where $u(t)$ solves \eqref{e7} and $w_a(t)$ solves \eqref{e4} with
$a=a(t)$. By the assumption, $w_a(t)$ exists for every $t\geq 0$.
The local existence of $u(t)$, the solution to \eqref{e7}, is the
conclusion of \lemref{lem2}. Let $\psi(t)\in C^1([0,\infty);X)$. In
the following lemma a proof of local existence of the solution to
problem \eqref{e7} is given by a novel argument. The right-hand side
of \eqref{e7} is a nonlinear function of $u$, which does not, in
general, satisfy the Lipschitz condition. This condition  is the
standard condition in the usual proofs of the local existence of the
solution to an evolution problem. Our argument uses an abstract
inverse function theorem. {\it This argument is valid under the
minimal assumption that $F'(u)$ depends continuously on $u$}.
\begin{lem}\label{lem2} If assumption
\eqref{e3} holds and \eqref{e4} is uniquely solvable for any $f\in X$,
then the solution $u(t)$ to \eqref{e7}
exists locally. \end{lem}
\begin{proof} Differentiate equation \eqref{e4} with $a=a(t)$
with respect to $t$. The result is \be\label{e16}
A_{a(t)}(w_a(t))\dot{w}_a(t)=-\dot{a}(t)w_a(t), \ee or \be\label{e17}
\dot{w}_a(t)=-\dot{a}(t)A^{-1}_{a(t)}(w_a(t))w_a(t). \ee Denote
\be\label{e18} \psi(t):=F(u(t))+a(t)u(t)-f. \ee
For any $\psi\in H$
equation \eqref{e18} is uniquely solvable for $u(t)$ by our assumption
\eqref{e4}, which is used with $f+\psi(t)$ in place of $f$ in \eqref{e4}.
By the inverse function theorem, which holds due to our assumption
\eqref{e3}, and by assumption \eqref{e2},  the solution
$u(t)$ to
\eqref{e18} is continuously differentiable with respect to $t$ provided
$\psi(t)$ is.
One may solve \eqref{e18} for $u$ and write $u=G(\psi)$, where
the map $G$ is continuously Fr\'{e}chet differentiable because $F$ is.

Differentiate \eqref{e18} and get \be\label{e19}
\dot{\psi}(t)=A_{a(t)}(u(t))\dot{u}(t)+\dot{a}(t)u. \ee If one wants the
solution to \eqref{e18} to be a solution to \eqref{e7}, then one has to
require that \be\label{e20} A_{a(t)}(u(t))\dot{u}=-\psi(t). \ee If
\eqref{e20} holds, then \eqref{e19} can be written as \be\label{e21}
\dot{\psi}(t)=-\psi+\dot{a}(t)G(\psi),\quad G(\psi):=u(t), \ee where
$G(\psi)$ is continuously Fr\'{e}chet differentiable. Thus, equation
\eqref{e21} is equivalent to \eqref{e7} at all $t\geq 0$ if \be\label{e22}
\psi(0)=F(u_0)+a(0)u_0-f. \ee Indeed, if $u$ solves \eqref{e7} then
$\psi$, defined in \eqref{e18}, solves the Cauchy problem
\eqref{e21}-\eqref{e22}.
Conversely, if $\psi$ solves \eqref{e21}-\eqref{e22}, then $u(t)$, defined
as the unique solution to \eqref{e18}, solves \eqref{e7}. Since the
right-hand side of \eqref{e21} is Fr\'{e}chet differentiable, it satisfies
a local Lipschitz condition. Thus, problem \eqref{e21}-\eqref{e22} is
locally, solvable. Therefore, problem \eqref{e7} is locally solvable.\\
\lemref{lem2} is proved. \end{proof}
It is known (see, for example, \cite{R499}) that the solution
$u(t)$ to
\eqref{e7} exists globally if the following
estimate holds: \be\label{e23} \sup_{t\geq 0}\|u(t)\|<\infty. \ee
\begin{lem}\label{lem3}
Estimate \eqref{e23} holds. \end{lem}
\begin{proof} Denote
\be\label{e24} z(t):=u(t)-w(t), \ee
where $u(t)$ solves
\eqref{e7} and $w(t)=w_{a(t)}$ solves \eqref{e4} with $a=a(t)$.
When $t\to \infty$, the function $w(t)$ tends to the limit $y$ by
\eqref{e5}, and, therefore, is uniformly bounded.
If one proves that
\be\label{e25} \lim_{t\to \infty}\|z(t)\|=0, \ee then \eqref{e23} follows
from \eqref{e25} and the boundedness of $w(t)$. Indeed,
\be\label{e26} \sup_{t\geq 0}\|u(t)\|\leq
\sup_{t\geq 0}\|z(t)\|+\sup_{t\geq 0}\|w(t)\|<\infty. \ee \end{proof}
{\it To prove \eqref{e25} we use \lemref{lem1}.}

Rewrite \eqref{e7} as
\be\label{e28}
\dot{z}=-\dot{w}-A_{a(t)}^{-1}(u(t))[F(u(t))-F(w(t))+a(t)z(t)]. \ee

\begin{lem}\label{lem4}
If the norm $\|w(t)\|$ in $X$ is differentiable, then
\be\label{e30}
\left|\frac {d\|w(t)\|}{dt}\right|\leq \|\dot{w}(t)\|.
\ee
\end{lem}
\begin{proof}
The triangle inequality implies:
\be\label{e31}
\frac{\|w(t+s)\|-\|w(t)\|}{s}\leq \frac{\|w(t+s)-w(t)\|}{s},\quad s>0.
\ee
Passing to the limit $s\searrow 0$ and using the assumption
concerning the differentiability of the norm in $X$, one gets
$ \frac {d\|w(t)\|}{dt}\leq \|\dot{w}(t)\|$. Similarly,
one gets $-\frac {d\|w(t)\|}{dt}\leq \|\dot{w}(t)\|$. These two
inequalities yield \eqref{e30}.\\
\lemref{lem4} is proved.
\end{proof}

Various necessary and sufficient conditions for the Gateaux or Fr\'echet
differentiability of the norm in Banach spaces
are known in the literature (see, for example, \cite{D} and \cite{DS}),
starting
with Shmulian's paper of 1940, see \cite{Sh}.

Hilbert spaces, $L^p(D)$ and $\ell^p$-spaces, $p\in (1, \infty)$,   and
Sobolev spaces $W^{\ell,p}(D)$, $p\in (1, \infty)$, $D\subset \R^n$
is a bounded domain, have Fr\'echet
differentiable norms. These spaces are uniformly convex and they have the
following property:  if $u_n\rightharpoonup u$ and $||u_n||\to ||u||$
as $n\to \infty$,  then $\lim_{n\to \infty}||u_n-u||=0$.

From \eqref{e17} and \eqref{en} one gets
\be\label{e32} \|\dot{w}\|\leq
c_1|\dot{a}(t)|r^{-b}(t)\|w(t)\|,\quad r(t)=|a(t)|, \ee
where $w(t):=w_a(t)$.
Since we assume that $\lim_{t\to
\infty}|a(t)|=0$, one concludes that \eqref{e5} and \eqref{e32}
imply the following inequality:
\be\label{e33}
\|\dot{w}\|\leq c_2|\dot{a}(t)|r^{-b}(t),\quad c_2=const>0, \ee because
\eqref{e5} implies the following estimate:
\be\label{e34} c_1\|w(t)\|\leq
c_2,\quad t\geq 0. \ee
Inequalities \eqref{en} and  \eqref{e33} imply that
\be\label{e35}
\|\dot{w}\|\leq c_2|\dot{r}(t)|r^{-b}(t),\quad t\geq 0. \ee
Recall that $F'(u):=A(u)$ and note that
\be\label{e36}
F(u)-F(w)=\int_0^1 F'(w+sz)ds z=A(u)z+\int_0^1[A(w+sz)-A(u)]ds z. \ee
Thus, one can write \eqref{e28} as
\be\label{e28'}
\dot{z}(t)=-z(t)-\dot{w}(t)-A^{-1}_{a(t)}(u(t))\eta (t):=-z(t)+W,
\ee
\be\label{e28a}\|\eta(t)\|=O(g^p(t)),\quad p=1+\kappa,
\quad g(t):=\|z(t)\|,
\ee
where estimate \eqref{e2}
was used, and $W$ is defined by the formula
\be\label{e28b}
W:=-\dot{w}(t)-A^{-1}_{a(t)}(u(t))\eta (t).
\ee

Let
\be\label{e28c}
Z(t):=e^{t}z(t).
\ee
Then \eqref{e28'} yields
\be\label{e37a}
e^{-t}\dot{Z}=W.
\ee
 Taking the norm of this equation yields
\be\label{e37b}
e^{-t}\|\dot{Z}\|=\|W\|.
\ee
One has
\be\label{e37c}\|W\|\le c_2|\dot{r}(t)|r^{-b}(t)+c_3r^{-b}(t)g^p(t),\qquad
g(t):=\|z(t)\|, \quad p=1+\kappa,
\ee
where $c_3:=c_0c_1$, $c_0$ is the constant from \eqref{e2} and $c_1$ is
the constant from \eqref{e3}.
Using estimate \eqref{e30}, one gets
\be\label{e37d}\|\dot{Z}\|\ge \left| \frac {d\|Z(t)\|}{dt}
\right|=\left|\frac
{d(e^tg(t))}{dt}\right|.\ee
Using formulas \eqref{e37a}-\eqref{e37d} one gets from \eqref{e28'}
the following inequality:
\be\label{e37} \dot{g}(t)\leq -g+
c_2|\dot{r}(t)|r^{-b}(t)+c_3r^{-b}(t)g^p, \qquad g(t)=\|z(t)\|, \quad
p=1+\kappa.
\ee

Inequality \eqref{e37} is of the form \eqref{e8} with \be\label{e39}
\gamma(t)=1,\quad \alpha(t)=c_3r^{-b}(t),\quad
\beta(t)=c_2|\dot{r}(t)|r^{-b}(t). \ee
Choose
\be\label{e40}\mu(t)=\lambda
r^{-k}(t),\quad \lambda=const>0,\quad k=const>0 .\ee
Then
\be\label{e41}
\dot{\mu}\mu^{-1}=-k\dot{r}r^{-1}. \ee
Let us assume that, as $t\to \infty$,
\be\label{e42}
r(t)\searrow 0,\quad \dot{r}<0,\quad |\dot{r}|\searrow 0. \ee
Assumption \eqref{e10} implies
\be\label{e43} g(0)\frac{\lambda}{r^k(0)}<1, \ee
and inequality \eqref{e9} holds if
\be\label{e44}
\frac{c_3r^{-b}(t)r^{kp}}{\lambda^p}+c_2|\dot{r}(t)|r^{-b}(t)\leq
\frac{r^k(t)}{\lambda}\big(1-k|\dot{r}(t)|r^{-1}(t)\big),\qquad t\geq 0.
\ee
Inequality
\eqref{e44} can be written as
\be\label{e45}
\frac{c_3r^{k(p-1)-b}(t)}{\lambda^{p-1}}+\frac{c_2\lambda|\dot{r}(t)|}
{r^{k+b}(t)}+\frac{k|\dot{r}(t)|}{r(t)}\leq
1, \qquad t\geq 0.
\ee

Let us choose $k$ so that
\bee
k(p-1)-b=1, \eee
that is,
\be\label{e46}
k=\frac{b+1}{p-1}.
\ee
Choose $\lambda$, for example, as follows:
\be\label{e47}
\lambda:=\frac{r^k(0)}{2g(0)}.
\ee
Then inequality \eqref{e43} holds, and
inequality \eqref{e45} can be
written as:
\be\label{e48}
c_3\frac{r(t)[2g(0)]^{p-1}}{[r^k(0)]^{p-1}}+c_2\frac{r^k(0)}{2g(0)}
\frac{|\dot{r}(t)|}{r^{k+b}(t)}+k\frac{|\dot{r}(t)|}{r(t)}\leq 1,\qquad
t\geq 0.
\ee
Note that \eqref{e46} implies:
\be\label{e49}
k+b=kp-1.
\ee
Choose $r(t)$ so that relations \eqref{e42} hold and
\be\label{e50}
k\frac{|\dot{r}(t)|}{r(t)}\leq \frac{1}{2},\qquad t\geq 0.
\ee
Since $r(0)\geq r(t)$ and \eqref{e50} holds, then inequality \eqref{e48}
holds if
\be\label{e51}
c_3\frac{[2g(0)]^{p-1}}{r^{b}(0)}+c_2\frac{r^k(0)}{2g(0)}
\frac{|\dot{r}(t)|}{r^{kp-1}}\leq \frac{1}{2},\qquad t\geq 0.
\ee
Denote
\be\label{e52}
c_2\frac{r^k(0)}{2g(0)}=c_2\lambda:=c_4.
\ee
Let
\be\label{e53}
c_4\frac{|\dot{r}(t)|}{r^{kp-1}}=\frac{1}{4},\qquad t\geq 0,
\ee
and $kp>2$. Then equation \eqref{e53} implies
\be\label{e54}
r(t)=\left[c_5+c_6t \right]^{-\frac{1}{kp-2}},\quad c_5=r^{2-kp}(0),\quad
c_6=\frac{kp-2}{4c_4},
\ee
where $c_5$ and $c_6$ are positive constants. Their explicit values
are not used below.
This $r(t)$ satisfies conditions \eqref{e42}, and equation \eqref{e53}
can be rewritten as:
\be\label{e55}
k\frac{|\dot{r}(t)|}{r(t)}=\frac{kr^{kp-2}(t)}{4c_4},\quad t\geq 0.
\ee
Recall that $r(t)$ decays monotonically. Therefore, inequality \eqref{e50}
holds if
\be\label{e56}
\frac{kr^{kp-2}(0)}{4c_4}\leq \frac{1}{2}.
\ee
Inequality \eqref{e56} holds if
\be\label{e57}
\frac{kg(0)}{c_2}r^{k(p-1)-2}(0)=\frac{kg(0)}{c_2}r^{b-1}(0)\leq 1,
\ee
because \eqref{e46} implies:
\be\label{e58}
k(p-1)-2=b-1.
\ee
Condition \eqref{e57} holds if $g(0)$ is sufficiently small or $r^{b-1}(0)$
is sufficiently large:
\be\label{e59}g(0)\leq \frac{c_2}{k}r^{b-1}(0).\ee

If $b>1$, then condition \eqref{e59} holds for any fixed $g(0)$ if $r(0)$
is sufficiently large. If $b=1$, then \eqref{e59} holds if $g(0)\leq
\frac{c_2}{k}$. If $b\in (0,1)$ then \eqref{e59} holds either if $g(0)$ is
sufficiently small or $r(0)$ is sufficiently small.

If \eqref{e54} and \eqref{e59} hold, then \eqref{e53} holds.
Consequently,
\eqref{e51} holds if
\be\label{e60}
c_3\frac{[2g(0)]^{p-1}}{r^{b}(0)}\leq \frac{1}{4}. \ee It follows from
\eqref{e59} that \eqref{e60} holds if
\be\label{e61}
c_32^{p-1}\left(\frac{c_2}{k}\right)^{p-1}\frac{1}{r^{-1+p+2b-bp}(0)}\leq
\frac{1}{4}. \ee
One has $p=1+\kappa$, and $\kappa\in (0,1]$. If $b>0$ and $\kappa\in
(0,1]$, then
\be\label{e62} -1+p-pb+2b=\kappa +(1-\kappa)b>0. \ee
Thus,
\eqref{e61} always holds if $r(0)$ is sufficiently large, specifically, if
\be\label{e63} r(0)\geq
[4c_3\left(2c_2 k^{-1}\right)^{p-1}]^{\frac{1}{\kappa  +(1-\kappa)b}}. \ee

We have proved the following theorem.
\begin{thm}\label{thm1} Let the Assumptions {\bf A1, A,2, and A3} hold.
 If $r(t)=|a(t)|$ is
defined in \eqref{e54}, and inequalities \eqref{e59} and \eqref{e63} hold,
then
\be\label{e64} \|z(t)\|<r^k(t)\lambda^{-1},\qquad \lim_{t\to
\infty}\|z(t)\|=0. \ee
Thus, problem \eqref{e7} has a unique global
solution $u(t)$ and
\be\label{e65} \lim_{t\to \infty}\|u(t)-y\|=0,
\ee
where $F(y)=f$.
\end{thm}

{\bf Proof of Lemma 1}.
Inequality \eqref{e8} can be written as
\be\label{e66}
-\gamma (t)\mu^{-1}(t)+\alpha(t, \mu^{-1}(t))+\beta(t)\le \frac
{d\mu^{-1}(t)}{dt}.
\ee
Let $\phi(t)$ solve the following Cauchy problem:
\be\label{e67}
\dot{\phi}(t)=-\gamma (t)\phi(t)+\alpha(t, \phi(t))+\beta(t),\quad t\ge 0,
\quad \phi(0)=\phi_0.
\ee
The assumption that $\alpha(t,g)$ is locally Lipschitz with
respect to $g$ guarantees local
existence and uniqueness of the solution $\phi (t)$ to problem
\eqref{e67}.
From the known comparison result (see, for instance, \cite{H}, Theorema
III.4.1)
it follows that
\be\label{e68}
\phi(t)\le \mu^{-1}(t) \qquad \forall t\ge 0,
\end{equation}
provided that $\phi(0)\le \mu^{-1}(0)$, where $\phi(t)$ is the unique
solution to problem \eqref{e67}. Let us take $\phi(0)=g(0)$.
Then  $\phi(0)\le \mu^{-1}(0)$ by the assumption in Lemma 1, and
inequality \eqref{e8} implies that
\be\label{e69}
g(t)\le \phi(t)    \qquad t\in [0,T).
\ee
Inequalities  $\phi(0)\le \mu^{-1}(0)$, \eqref{e68},
and \eqref{e69} imply
\be\label{e70}
g(t)\le \phi(t)\le \mu^{-1}(t), \qquad t\in [0,T).
\ee
By the assumption, the function $\mu(t)$ is defined for all $t\ge 0$
and is bounded on any compact subinterval of the set $[0,\infty)$.
Consequently, the functions $\phi(t)$ and $g(t)\ge 0$ are defined
for all $t\ge 0$, and estimate \eqref{e11} is established.
Lemma 1 is proved. \hfill $\Box$

When this paper was under consideration, convergence of the
DSM for general operator equations was established in \cite{R627}.

{\bf Acknowledgement} This paper was completed during a visit to MPI
for mathematics, Leipzig. The author thanks MPI for hospitality.


\newpage

\begin{thebibliography}{00}

\bibitem{DK} Yu. L. Daleckii,  M. G. Krein, {\it Stability of solutions of
differential equations in Banach spaces,} Amer. Math. Soc.,
Providence, RI, 1974.

\bibitem{D} M.Day, {\it Normed linear spaces}, Springer-Verlag, Berlin,
1958.
\bibitem{DS} N.Dunford, J.Schwartz, {\it Linear operators},
Part 1: General theory, Interscience, New York, 1958.

\bibitem{G} M. Gavurin, Nonlinear functional equations and continuous
analysis of iterative methods, Izvestiya VUS'ov, Mathem., 5, (1958), 18-31
(in Russian)

\bibitem{H} P. Hartman,{\it  Ordinary differential equations}, J.Wiley,
New York, 1964.


\bibitem{R574} N.S.Hoang and A.G.Ramm,
 Dynamical Systems Method for solving
nonlinear equations with monotone operators,
Math. of Comput., 79, 269, (2010), 239-258.

\bibitem{R596} N.S.Hoang and A.G.Ramm,
Nonlinear differential inequality,
Mathematical Inequalities and Applications (MIA), 14, N4, (2011), 967-976.

\bibitem{R604} N.S.Hoang and A.G.Ramm, Some nonlinear inequalities and
applications,
Journ. of Abstract Diff. Equations and Applications,
2, N1, (2011), 84-101.

\bibitem{R612} N.S.Hoang and A.G.Ramm,
 \textit{ Dynamical Systems Method and Applications.
Theoretical Developments and Numerical Examples.}
Wiley, Hoboken, 2012.

\bibitem{R499} A.G. Ramm, \textit{ Dynamical systems method for solving
operator equations}, Elsevier, Amsterdam, 2007.

\bibitem{R485} A.G. Ramm, Dynamical systems method (DSM) and
nonlinear problems, in the book: Spectral Theory and Nonlinear
Analysis,
World Scientific Publishers, Singapore, 2005, 201-228. (ed J.
Lopez-Gomez).

\bibitem{R601} A.G. Ramm, On the DSM version of Newton's method,
 Eurasian Math. Journ (EMJ), 2, N3, (2011), 91-99.

\bibitem{R570} A.G. Ramm, Asymptotic stability of solutions to abstract
differential equations, Journ. of Abstract Diff. Equations and
Applications (JADEA), 1, N1, (2010), 27-34.


\bibitem{R605} A.G. Ramm, Stability of solutions to some evolution
problems, Chaotic Modeling and Simulation (CMSIM), 1, (2011), 17-27.



\bibitem{R611} A.G. Ramm, On the DSM Newton-type method, J. Appl.Math.
and Comp.,(JAMC), 38, N1-2, (2012), 523-533.


\bibitem{R606} A.G. Ramm, How large is the class of operator equations
solvable
by a DSM Newton-type method ? Appl. Math. Lett, 24, N6, (2011), 860-865.

\bibitem{Sh} V.Shmulian, On differentiability of the norm in Banach space,
Doklady Acad. Sci. USSR, 27, (1940), 643-648.

\bibitem{T} R. Temam, {\it Infinite-dimensional dynamical systems in
mechanics and physics}, Springer-Verlag, New York, 1997.

\bibitem{R627} A.G. Ramm, DSM for general nonlinear equations, Appl.Math. 
Lett., (2012) http://dx.doi.org/10.1016/j.aml.2012.06.006




\end{thebibliography}
\end{document}